\documentclass[letterpaper, 10 pt, conference]{ieeeconf}  % Comment this line out if you need a4paper

\IEEEoverridecommandlockouts                              % This command is only needed if 
\usepackage{amsmath,amssymb,color}
\newtheorem{theorem}{Theorem} 
\newtheorem{algorithm}{Algorithm}

\newtheorem{lemma}{Lemma}
\newtheorem{remark}{Remark}
\newtheorem{example}{Example}
\newtheorem{definition}{Definition} 
\newtheorem{assumption}{Assumption}
 \usepackage{algorithm}
 \usepackage{graphicx}
  \usepackage{algorithmic}
\usepackage{float}                  % 图片浮动位置
\usepackage{subfig}                 % 子图包，不要与{subfigure}混用，{subfig}较新
\usepackage{overpic}  

\usepackage{mathptmx} 
\DeclareMathOperator{\diag}{diag}

\title{\LARGE \bf
Online Learning for Nonlinear Dynamical Systems \\without the I.I.D. Condition}

\author{Lantian Zhang$^{1}$ and Silun Zhang$^{2}$% <-this % stops a space
\thanks{*This work is supported by The Wallenberg AI, Autonomous Systems and Software Program (WASP) which is funded by the Knut and Alice Wallenberg Foundation.}% <-this % stops a space
\thanks{$^{1}$ Lantian Zhang is with the Department of Mathematics, KTH Royal Institute of Technology, 100 44 Stockholm, Sweden. {\tt\small lantian@kth.se}}%
\thanks{$^{2}$ Silun Zhang is with the Department of Mathematics, KTH Royal Institute of Technology, 100 44 Stockholm, Sweden. {\tt\small silunz@kth.se}}%
}

\begin{document}

\maketitle
\thispagestyle{empty}
\pagestyle{empty}

%%%%%%%%%%%%%%%%%%%%%%%%%%%%%%%%%%%%%%%%%%%%%%%%%%%%%%%%%%%%%%%%%%%%%%%%%%%%%%%%
\begin{abstract}
This paper investigates online identification and prediction for nonlinear stochastic dynamical systems. 
In contrast to offline learning methods, we develop online algorithms that learn unknown parameters from a single trajectory. A key challenge in this setting is handling the non-independent data generated by the closed-loop system. Existing theoretical guarantees for such systems are mostly restricted to the assumption that inputs are independently and identically distributed (i.i.d.), or that the closed-loop data satisfy a persistent excitation (PE)  condition. However, these assumptions are often violated in applications such as adaptive feedback control. 
In this paper, we propose an online projected Newton-type algorithm for parameter estimation in nonlinear stochastic dynamical systems, and develop an online predictor for system outputs based on online parameter estimates. By using both the stochastic Lyapunov function and martingale estimation
methods, we demonstrate that the average regret converges to zero without requiring traditional persistent excitation (PE) conditions. Furthermore, we establish a novel excitation condition that ensures global convergence of the online parameter estimates. The proposed excitation condition is applicable to a broader class of system trajectories,  including those violating the PE condition.

\end{abstract}

%%%%%%%%%%%%%%%%%%%%%%%%%%%%%%%%%%%%%%%%%%%%%%%%%%%%%%%%%%%%%%%%%%%%%%%%%%%%%%%%
\section{INTRODUCTION}

Dynamical systems are fundamental for modeling a wide range of problems arising in complex physical phenomena, cyber-physical infrastructures, and machine learning tasks. Recent neural network architectures for sequential data processing—such as recurrent neural networks and LSTMs—can be viewed as instances of nonlinear dynamical systems. 
In many applications, the system dynamics are unknown and potentially time-varying.  A fundamental problem in control systems and machine learning is how to identify these unknown dynamics from past time series of trajectories and provide online predictions of the future behavior of the dynamical system.
This paper considers the online identification and prediction problem for nonlinear stochastic dynamical systems of the form
 \begin{equation}\label{e1}
\begin{aligned}
			x_{t+1}&=h(\theta^{*}, x_{t},  u_{t})+w_{t+1},\\
            y_{t+1}&=x_{t+1},
\end{aligned}
 \end{equation}
where $y_{t}$, $x_{t}$, $u_{t}$, $w_{t}$ denote the system outputs, states, inputs, and noise, respectively, and $\theta$ is the unknown parameter to be learned.

Many studies have investigated the learning of unknown parameter $\theta^{*}$ in system (\ref{e1}) using batch learning methods (\cite{allen, caysi}),
which are offline approaches and require the data collection of multiple independent system trajectories. In contrast, this paper proposes an online learning method that learns unknown parameters from a single trajectory $\{(u_{t}, y_{t+1}), t\geq 0\}$. There is some existing work on online learning of nonlinear systems, particularly when the system dynamics include saturation nonlinearities. For example, \cite{yx2022} established the convergence results when the control input is absent and the observed state sequence is required to be an irreducible and aperiodic Markov chain. \cite{oymak} analyzed the SGD algorithm with asymptotic convergence under an i.i.d. Gaussian input sequence $\{u_{t}, t\geq 0\}$. \cite{bahmani}  generalized the results in \cite{oymak} and relaxed the input requirement from i.i.d. Gaussian to any i.i.d. distributions with a heavier-tailed inputs. Moreover, for the general nonlinear dynamics, \cite{regret} considered non-linear autoregressive models, and established the asymptotic upper bound for regrets under strong persistent excitation (PE) condition: $\frac{1}{T}\sum\limits_{t=1}^{T}\nabla_{\theta}h(\theta, x_{t}, u_{t})\nabla_{\theta}^{\top}h(\theta, x_{t}, u_{t})$ converges (as $T\to \infty$) to a positive definite matrix a.s. Furthermore, \cite{sattar} established convergence results for parameter estimation under the assumption that a pre-designed stabilizing controller is known a priori. The controller follows the structure: $u_{t}=\pi(x_{t})+z_{t}$, where $\pi(\cdot)$ is a state feedback controller that exponentially stabilizes the system, and $z_{t}$ is i.i.d. excitation noise. Due to the exponential decay of variable dependencies in stable systems, the work demonstrates that closed-loop data can be used to generate approximately i.i.d. trajectories, thereby preserving the PE condition. In summary, almost all of the existing identification results for nonlinear stochastic dynamical systems
need at least the usual PE condition on the system data, and
actually, most need i.i.d assumptions. Though these idealized
conditions are convenient for theoretical investigation, they are
hardly satisfied or verified in many stochastic dynamical
systems with feedback signals (see, e.g. \cite{feed}).

A relevant field that addresses the online learning problem under general data conditions is system identification and adaptive control, where extensive studies have solved the problem for linear dynamical systems (\cite{str, ljung1976, moore1978, chris, lw1982, chen1982, goodwin, lai1986, kumar, g1991, g1995, cg1991}). Among these works, we mention that Lai and Wei \cite{lw1982} developed an asymptotic analysis of Least Squares (LS) under the weakest possible non-PE condition on signals, which is achieved by using stochastic Lyapunov functions and martingale convergence theorems. Moreover, variety of learning methods have been introduced for online learning problems of linear stochastic dynamical systems with nonlinear observations, including empirical measure approaches (\cite{ZG2003}), set-membership methods (\cite{MC2011}), maximum likelihood algorithms (\cite{godoy}), and stochastic approximation-type algorithms (\cite{GZ2013, you2015}). Recently, in \cite{ZZ2022, zztsqn}, the strong consistency of estimators was established under general non-PE
conditions, which is similar to
the weakest possible signal condition for a stochastic linear
regression model with regular observations. While all these studies focus on linear dynamical systems, the analytical techniques therein inspire the present study for the nonlinear stochastic system $(\ref{e1})$.
%Besides,  
%a 

%

Based on our previous research on online learning under nonlinear observations in \cite{ZZ2022, zztsqn}, this paper investigates the problem of online identification and prediction for nonlinear dynamical systems. The main difference from our earlier work lies in the nature of the system: the previous study focused on linear dynamical systems with nonlinear observation models, whereas the current work addresses nonlinear dynamical systems. Moreover, the nonlinearities considered in \cite{ZZ2022, zztsqn} were limited to saturation-type functions, while the present formulation allows for a broader class of nonlinear systems. In this paper, for the nonlinear system (\ref{e1}), we establish the convergence of the average adaptive prediction regret and provide an asymptotic upper bound on the parameter estimation error. Our assumptions on the loss function are weaker than the restricted secant inequality and allow for non-convex optimization landscapes. Specifically, we make the following contributions:
\begin{itemize}
    \item We propose a novel online learning algorithm to estimate the unknown parameters of the nonlinear dynamical system $(\ref{e1})$, which allows for non-convex optimization landscape.
    \item We prove the global convergence for the average regret almost surely by leveraging a Lyapunov function and martingale convergence techniques. Compared with the related work of \cite{regret}, our theoretical results do not require the PE conditions.

    \item We establish a new excitation condition on data to guarantee global convergence of online parameter estimates. Unlike prior work of \cite{sattar} that relies on persistently excited closed-loop data, our data conditions can apply to a broader class of closed-loop trajectories, including those without PE conditions. 
\end{itemize}

\section{PROBLEM FORMULATION}
{\bf Notations.}  By $\|\cdot\|$, we denote the 2-norm of vectors or matrices. The maximum and minimum eigenvalues of a symmetric matrix $M$ are denoted by $\lambda_{max}\left\{M\right\}$ and $\lambda_{min}\left\{M\right\}$,  respectively. Denote $\det (M)$ (or $|M|$ in derivation) the determinant of a matrix $M$. Furthermore, let $\{\mathcal{F}_{t}, t\geq 0\}$ be a sequence of non-decreasing $\sigma-$algebras.

Consider the discrete-time stochastic nonlinear dynamical system defined $\forall t\geq 0$:
\begin{equation}\label{system}
\begin{aligned}
			x_{t+1}&=h(\theta^{*}, x_{t}, u_{t})+w_{t+1},\\
            y_{t+1}&=x_{t+1},  
\end{aligned}
\end{equation}
where $\theta^{*} \in \mathbb{R}^{p},$ $y_{t}\in \mathbb{R}^{n}$, $x_{t}\in \mathbb{R}^{n}$, $u_{t}\in \mathbb{R}^{m}$ and $w_{t}\in \mathbb{R}^{n}$ representing the unknown parameter, system output, system state, system input and system noise, respectively; $h(\cdot)=(h_{1}(\cdot),\cdots,h_{n}(\cdot)):  \mathbb{R}^{n+m+p}\rightarrow\mathbb{R}^{n}$ is a nonlinear function.

In online learning scenarios, a crucial challenge in estimating unknown parameters from a data stream is designing effective learning algorithms capable of handling non-independent data. This challenge is more striking when data is collected from a running closed-loop system, where the system input depends on the states and is generated by a designed feedback controller. In this paper, we assume the system is driven by inputs 
\begin{equation}\label{con}
u_{t}=\pi(x_{t})+\epsilon_{t},
\end{equation}
where $\pi(\cdot)$ is a unknown fixed control policy, $\epsilon_{t}$ is an ``exploration signal''. Besides, we assume that the gradient sequence $\left\{\nabla_{\theta}h(\theta^{*}, x_{t}, u_{t})\in \mathbb{R}^{n\times p},  \mathcal{F}_{t}\right\}$ is an adapted sequence, where $\{\mathcal{F}_{t}, t\geq 0\}$ is a sequence of non-decreasing $\sigma-$algebra. From $(\ref{system})$, the optimal predictor of the output $y_{t+1}$ based on the current and past observations $\{y_{t}, u_{t},\cdots,y_{1}, u_{1}\}$ for each $t\geq 2$ is $\bar{y}_{t+1}$, defined by 
\begin{equation}\label{eq2}
\bar{y}_{t+1}=\mathbb{E}\left[y_{t+1}\mid \mathcal{F}_{t}\right]=h(\theta^{*}, y_{t}, u_{t}), 
\end{equation}
with the optimal mean squares prediction error 
\begin{equation}
\mathbb{E}\left[\|y_{t+1}-\bar{y}_{t+1}\|^{2}\right]=\mathbb{E}\left[\|w_{t+1}\|^{2}\right].
\end{equation}
However, in practice, the parameter $\theta^{*}$ is usually unknown. A natural way to give the online predictor for each $t\geq 1$, is to first give the estimates $\hat{\theta}_{t}$ for the unknown parameter based on the current and past observations $\{y_{t}, u_{t},\cdots,y_{1}, u_{1}\}$, and then replace $\theta^{*}$ by $\hat{\theta}_{t}$ in $(\ref{eq2})$. This leads to the online predictor
\begin{equation}
\hat{y}_{t+1}=h(\hat{\theta}_{t}, y_{t}, u_{t}). 
\end{equation}
Usually, the accumulated squared difference between the optimal predictor $\bar{y}_{t+1}$ and the adaptive predictor $\hat{y}_{t+1}$ is called regret $R_{T}$, defined by 
\begin{equation}\label{eqre}
\begin{aligned}
R_{T}&=\sum_{t=0}^{T-1}\left\|\bar{y}_{t+1}-\hat{y}_{t+1}\right\|^{2}\\
&=\sum_{t=0}^{T-1}\left\|h(\theta^{*}, y_{t}, u_{t})-h(\hat{\theta}_{t}, y_{t}, u_{t})\right\|^{2}
\end{aligned}
\end{equation}
The main objectives of this work are twofold: (i) to design an online algorithm that recursively estimates the parameter $\theta^{*}$ through streaming data, generating estimates $\hat{\theta}_{t}$ at each time step t, and (ii) to establish the almost sure asymptotic guarantee that the cumulative regret satisfies
$$R_{T}=o(T),\;\;a.s.$$

For the adaptive algorithm design, we consider the loss function at each time step $t \geq 0$ as 
\begin{equation}\label{eq7}
\begin{aligned}
\mathcal{L}(\theta, y_{t}, u_{t})=\left\|h(\theta^{*}, y_{t}, u_{t})-h(\theta, y_{t}, u_{t})\right\|^{2}.
\end{aligned}
\end{equation}
We next introduce a definition concerning the stability of the closed-loop system.
\begin{definition}[$\rho-$stable]
The closed-loop system $(\ref{system})$ is defined to be  $\rho$-stable if there exists a constant $\rho\in (0,1)$ and a norm $\|\cdot\|_{1}$ on $\mathbb{R}^{n}$ such that, for all $t \geq 0$,
\begin{equation}\label{e11}
\|h(\theta^{*}, x_{t}, u_{t})\|_{1}\leq \rho \|x_{t}\|_{1}.
\end{equation}
\end{definition}

To proceed with the paper, we need the following assumptions:
 \begin{assumption}[Unknown parameter]\label{assum1}
     The unknown parameter $\theta^{*}$ belongs to a known bounded set, i.e., there exists a constant $D>0$ such that $\|\theta^{*}\|\leq D$.
 \end{assumption}

\begin{assumption}[Random noise]\label{assum2}
The noise $\{w_{t+1}, \mathcal{F}_{t}\}$ is a martingale difference sequence, and there exists a constant $\eta>2$ such that
\begin{equation}
\begin{aligned}
\sup_{t\geq 0}\mathbb{E}\left[\|w_{t+1}\|^{\eta}\mid \mathcal{F}_{t}\right]<\infty,\;a.s.
\end{aligned}
\end{equation}
\end{assumption}

\begin{assumption}[Control input]\label{assum3}
The control policy $\pi(\cdot)$ in (\ref{con}) is Lipschitz continuous, i.e., there exists a constant $L>0$ such that
\begin{equation}
\|\pi(x)-\pi(y)\|\leq L\|x-y\|, \forall x, y \in \mathbb{R}^{n}.
\end{equation}
The exploration signals $\{\epsilon_{t}, t\geq 0\}$ are uniformly bounded, that is, there exists a constant $\bar{\epsilon}>0$ such that$\|\epsilon_{t}\|\leq \bar{\epsilon}$ for all $t \geq 0$. Moreover, the closed-loop system $(\ref{system})$ is $\rho$-stable. 
\end{assumption}

\begin{assumption}[Loss function]\label{assum4}
 Given any $r>0$, there exists $\alpha(r)>0$ and $\beta\geq 1$ such that the loss function $\mathcal{L}(\theta, y_{t}, u_{t})$ in (\ref{eq7}) satisfies
 \begin{equation}\label{q9}
\begin{aligned}
\langle\theta-\theta^{*}, \nabla_{\theta} \mathcal{L}(\theta, y_{t}, u_{t})\rangle\geq \alpha(r)\|\nabla_{\theta}^{\top}h(\theta, y_{t}, u_{t})(\theta-\theta^{*})\|^{2},
\end{aligned}
\end{equation}
and
 \begin{equation}\label{eq10}
\begin{aligned}
\mathcal{L}(\theta, y_{t}, u_{t})\leq  \beta\langle\theta-\theta^{*}, \nabla \mathcal{L}_{\theta}(\theta, y_{t}, u_{t})\rangle,
\end{aligned}
\end{equation}
 for any $\|\theta-\theta^{*}\|\leq r$, $\|y_{t}\|\leq r$, $\|u_{t}\|\leq r$, and $t\geq 0$.
Besides, the gradient $\nabla_{\theta}h(\theta, y_{t}, u_{t})$ is continuous with respect to $\theta$ for all $t\geq 0$. Moreover, given any $r>0$, there exist a constant $M(r)>0$ such that for all $a_{1}, a_{2} \in \mathbb{R}^{n}$, $b_{1}, b_{2} \in \mathbb{R}^{m}$, and $\|\theta-\theta^{*}\|\leq r$, the gradient $\nabla_{\theta}h(\theta, \cdot, \cdot)$ satisfies
 \begin{equation}
 \begin{aligned}
\left\|\nabla_{\theta}h(\theta,a_{1}, b_{1})-\nabla_{\theta}h(\theta, a_{2}, b_{2})\right\|\leq M(r)\|\xi_{1}-\xi_{2}\|,\;\;a.s.,
\end{aligned}
 \end{equation} 
 where $\xi_{i}=[a_{i}^{\top}, b_{i}^{\top}]^{\top}, i=1,2.$
\end{assumption}
\begin{remark}
We remark that condition $(\ref{q9})$-$(\ref{eq10})$ is fairly mild than the traditional strong convexity assumption. The inequality $(\ref{q9})$-$(\ref{eq10})$ does not even imply the convexity
of $\mathcal{L}(\theta, y_{t}, u_{t})$. Condition $(\ref{q9})$ is weaker than the restricted secant inequality (\cite{karimi}), which was used in \cite{sattar} to analyze the learning performance of nonlinear dynamical systems. Below, we present three specific examples of systems that satisfy Assumption \ref{assum4}.
The first example is the classical linear system, where the loss function $\mathcal{L}(\theta,y_{t}, u_{t})$ is convex. The second example comes from the nonlinear dynamic equations in recurrent neural network (RNN) models, where the loss function $\mathcal{L}(\theta,y_{t}, u_{t})$ is not even convex.
The third example is the binary observation state model, and in this case, the loss function $\mathcal{L}(\theta,y_{t}, u_{t})$ also exhibits non-convexity.
\end{remark}

\begin{example}[Linear dynamics]\label{ex2}
Consider the classical linear stochastic system defined by
\begin{equation}
\begin{aligned}
x_{t+1}&=Ax_{t}+Bu_{t}+w_{t+1},\\
y_{t+1}&=x_{t+1},
\end{aligned}
\end{equation}
 where $A\in \mathbb{R}^{n\times n}$, $B\in \mathbb{R}^{n\times m}$ are the unknown parameters. Let $\theta=[A^{\top}_{1}, \cdots, A^{\top}_{n}, B^{\top}_{1},\cdots, B^{\top}_{n}]^{\top}$, where $A_i$ and $B_i$ denote the $i$-th row of the matrices $A$ and $B$, respectively, for $1 \leq i \leq n$. From (\ref{eq7}), the loss function $\mathcal{L}(\theta, y_{t}, u_{t})=\|(A^{*}-A)y_{t}+(B^{*}-B)u_{t}\|^{2}$. Then Assumption \ref{assum4} can be verified with $\alpha(r)\equiv 1$, $\beta=1,$
 and $M(r)\equiv 1$.
\end{example}

\begin{example}[RNN dynamics]\label{ex1}
Consider the system
\begin{equation}\label{e14}
x_{t+1}=\sigma(Ax_{t}+Bu_{t})+w_{t+1},
\end{equation}
where $A\in \mathbb{R}^{n\times n}$, $B\in \mathbb{R}^{n\times m}$ are the unknown system parameters, $\sigma(\tau_{1}, \cdots, \tau_{n})=[\sigma_{1}(\tau_{1}), \cdots,\sigma_{n}(\tau_{n})]^{\top}: \mathbb{R}^{n}\to \mathbb{R}^{n}$, where each $\sigma_i(\cdot)$ is the sigmoid function defined by
\begin{equation}\label{sigma}
\sigma_{i}(x)=\frac{1}{1+e^{-x}}.
\end{equation}
This dynamic is the backbone of classical Elman recurrent neural network (RNN) (see, e.g., \cite{elman, enn, rnn, garzon}), 
 which have broad applications in sequential
 learning tasks. Let $\theta=[A^{\top}_{1}, \cdots, A^{\top}_{n}, B^{\top}_{1},\cdots, B^{\top}_{n}]^{\top}$, where $A_i$ and $B_i$ denote the $i$-th row of the matrices $A$ and $B$, respectively, for $1 \leq i \leq n$. In this example, the loss function $\mathcal{L}(\theta, y_{t}, u_{t})=\|\sigma(A^{*}y_{t}+B^{*}u_{t})-\sigma(Ay_{t}+Bu_{t})\|^{2}$. With $\alpha(r) = 2\sigma_i(2r^2)$, $\beta = 1$, and $M(r) \equiv 1$, Assumption \ref{assum4} can also be verified to hold. A detailed proof is provided in Appendix.
\end{example}

\begin{example}[Binary-valued dynamics]\label{ex3}
Consider the binary-valued dynamical system as follows:
\begin{equation}\label{e17}
x_{t+1}=I(Ax_{t}+Bu_{t}+v_{t+1}>c_{t}),
\end{equation}
where $A\in \mathbb{R}^{n\times n}$, $B\in \mathbb{R}^{n\times m}$ are unknown system parameters, $\{c_{t}, \mathcal{F}_{t}\}$ is a bounded adapted threshold sequence, $v_{t+1}$ is the i.i.d. system noise with $v_{t+1}\sim N(0, I_{n\times n})$, $I(\cdot)$ is the indicator function and applies entry-wise on vector inputs. This kind of nonlinear dynamical systems emerges in a wide range of applications, such as neuronal dynamics (\cite{jd1989}), social interactions (\cite{brock, dur2010}), and stock and option price dynamics (\cite{sto}).

Define the nonlinear function $f(\cdot)$ on $\mathbb{R}^{n}$ as
\begin{equation}
f(x)=\mathbb{E}\left[I(x+v_{t+1}>c_{t})\mid \mathcal{F}_{t}\right],
\end{equation}
then system $(\ref{e17})$ is equivalent to 
\begin{equation}\label{e18}
x_{t+1}=f(Ax_{t}+Bu_{t})+w_{t+1},
\end{equation}
where $\{w_{t}, t\geq 1\}$ is a martingale difference sequence. Let $\theta=[A^{\top}_{1}, \cdots, A^{\top}_{n}, B^{\top}_{1},\cdots, B^{\top}_{n}]^{\top}$, where $A_i$ and $B_i$ denote the $i$-th row of the matrices $A$ and $B$, respectively, for $1 \leq i \leq n$. Here, the loss function $\mathcal{L}(\theta, y_{t}, u_{t})=\|f(A^{*}y_{t}+B^{*}u_{t})-f(Ay_{t}+Bu_{t})\|^{2}$. It can be verified that the system in \eqref{e18} satisfies Assumption \ref{assum4} with $\alpha(r) = 2f_{i}(2r^2)$, $\beta = 1$, and $M(r) \equiv 1$.
\end{example}

\section{MAIN RESULTS}
In this section, we first present a novel online estimation algorithm for unknown parameters in the nonlinear stochastic system (\ref{system}). Subsequently, we provide theoretical guarantees for the asymptotic convergence of both the online prediction regret and parameter estimation error.

Inspired by the construct form of the linear least squares (LS) algorithm, we propose an iterative algorithm  for the loss function  $\mathcal{L}(\theta, y_{t}, u_{t})$ defined in (\ref{eq7}), as given in Algorithm $1$.

\begin{algorithm}[h]\label{alg111}
 \caption{Adaptive estimation algorithm for unknown parameter $\theta^{*}$}
Begin with an arbitrarily initial estimate $\hat{\theta}_{0}\in \mathbb{R}^{p}$ and an initial positive definite matrix  $P_{0}=I_{p\times p}$. The estimates $\hat{\theta}_{t}$ for $\theta^{*}$ are recursively defined at any iteration $t\geq 0$ as follows:
 %(C^{-1}y_{t+1})^{(i)}
 \begin{equation}\label{eq6}
 \begin{aligned}
\hat{\theta}_{t+1}=&\Pi_{P_{t+1}^{-1}}\{\hat{\theta}_{t}+\eta_{t}P_{t+1}\nabla_{\theta}^{\top}h(\hat{\theta}_{t}, y_{t}, u_{t})\left[y_{t+1}-h\left(\hat{\theta}_{t}, y_{t}, u_{t}\right)\right]\},\\
 P_{t+1}=&P_{t}-\eta_{t}^{2}P_{t}\nabla_{\theta}^{\top}h(\hat{\theta}_{t}, y_{t}, u_{t})\Gamma_{t}\nabla_{\theta}h(\hat{\theta}_{t}, y_{t}, u_{t})P_{t},\\
\Gamma_{t}=&(I+\eta_{t}^{2}\nabla_{\theta}h(\hat{\theta}_{t}, y_{t}, u_{t})P_{t}\nabla_{\theta}^{\top}h(\hat{\theta}_{t}, y_{t}, u_{t}))^{-1},\\
\eta_{t}=&\frac{1}{\alpha^{-1}_{t}+2\beta\|\nabla_{\theta}h(\hat{\theta}_{t}, y_{t}, u_{t})P_{t+1}\nabla_{\theta}^{\top}h(\hat{\theta}_{t}, y_{t}, u_{t})\|},
 \end{aligned}
 \end{equation}
 where $\alpha_{t}=\alpha(\|y_{t}\|+\|u_{t}\|+D)$, $D$ is defined as in Assumption \ref{assum1}, and $\Pi_{P_{t+1}^{-1}}\{\cdot\}$ is the projection operator defined via a time-varying Frobenius norm as follows:
	\begin{equation}	
		 \Pi_{P_{t+1}^{-1}}\{x\}=\mathop{\arg\min}_{\|\omega\|\leq D}(x-\omega)^{\top}P_{t+1}^{-1}(x-\omega), \forall x \in \mathbb{R}^{p}.
		 \end{equation}		  
 \end{algorithm}

For Algorithm 1, we establish the following global convergence guarantees.

\begin{theorem}\label{thm2}
Under Assumptions $\ref{assum1}$-$\ref{assum4}$, the sample paths of the accumulated regrets will have the following property: 
 \begin{equation}\label{eqq9}
R_{T}=o\left(T\right),\;\;a.s., \;\; T\to \infty,
 \end{equation}
 where $R_{T}$ is defined in (\ref{eqre}). Moreover, if the system output sequence $\{y_{t}, t\geq 0\}$ is bounded, then we have
  \begin{equation}\label{eqq22}
R_{T}=O\left(\log T\right),\;\;a.s., \;\; T\to \infty.
 \end{equation}
 \end{theorem} 
 \begin{remark}\label{re77}
 (\ref{eqq9}) shows that as time tends to infinity, the average regret of the adaptive predictor converges to zero. This indicates that the proposed adaptive predictor can asymptotically approximate the optimal predictor in the mean square sense. Furthermore, if the noise is bounded, one can easily verify the boundedness of the output sequence $\{y_t, t \geq 0\}$, thereby ensuring the logarithmic regret (\ref{eqq22}). We remark that our results do not require the strong PE condition assumed in \cite{regret}, which demands almost sure (a.s.) convergence of the matrix $\frac{1}{T}\sum\limits_{t=1}^{T}\nabla_{\theta}h(\theta, y_{t}, u_{t})\nabla_{\theta}^{\top}h(\theta, y_{t}, u_{t})$ to a positive definite matrix.
\end{remark}

\begin{theorem}\label{thm1}
Under Assumptions $\ref{assum1}$-$\ref{assum4}$, the parameter estimation error given by Algorithm 1 has an asymptotic upper bound as $T\to \infty$, in specific,
\begin{equation}\label{6}
\left\|\theta^{*}-\hat{\theta}_{T}\right\|^{2}=O\left(\frac{\log T}{\lambda_{\min}\left\{P_{T+1}^{-1}\right\}}\right), a.s.
\end{equation}
 \end{theorem} 
 % Compare

\begin{remark}
From Theorem \ref{thm1}, it is easy to see that the algorithm will converge to the true parameter almost surely if 
	\begin{equation}\label{re}
		\log t=o\left(\lambda_{\min}\left\{P_{t+1}^{-1}\right\}\right),\;\;a.s.,\;\; t\rightarrow \infty.
		\end{equation}
The convergence condition $(\ref{re})$ specifies the excitation condition on the closed-loop system trajectory to identify the system parameters. This condition is weaker than the control requirements imposed in previous work of \cite{sattar}. The analysis in \cite{sattar} requires the excitation signal $\epsilon_{t}$ to be an independent and identically distributed (i.i.d.) sequence, which ensures satisfaction of the PE condition: $t = O\left(\lambda_{\min}\{P_t^{-1}\}\right)$ in high probability. Moreover, our excitation condition can be verified in non-PE scenarios. For example, in system (\ref{e14}) with bounded noise $\{w_{t}, t\geq 0\}$, if the control input is chosen as $$u_{t}=Ky_{t}+t^{-\frac{1}{4}}\sqrt{\log t}\tilde{\epsilon}_{t},$$ where $K$ is a matrix such that $A - BK$ is stable, and $\{\tilde{\epsilon}_t, t\geq 0\}$ is an i.i.d. bounded sequence with $\mathbb{E}[\tilde{\epsilon}_t] = 0$ and $\mathbb{E}[\tilde{\epsilon}_t \tilde{\epsilon}_t^{\top}] = I$ for each $t\geq 0$, then it can be shown that $\sqrt{t}\log t = O\left(\lambda_{\min}\{P_t^{-1}\}\right)$ (see, e.g., \cite{g1999}), which guarantees the convergence of the parameter estimate $\hat{\theta}_t$ to the true parameter $\theta^{*}$ almost surely.
 \end{remark}

\section{PROOF OF MAIN RESULTS}

In this section, we present the proof of Theorem \ref{thm2} and Theorem \ref{thm1}. For simplicity, for each $t\geq 0$, let
\begin{equation}\label{e27}
\begin{aligned}
\phi_{t}=&\nabla_{\theta}^{\top}h(\hat{\theta}_{t}, y_{t}, u_{t}),\;\; r_{t}=\sum_{i=0}^{t}\|\phi_{i}\|^{2},\\
\psi_{t}=&h\left(\theta^{*}, y_{t}, u_{t}\right)-h\left(\hat{\theta}_{t}, y_{t}, u_{t}\right).
\end{aligned}
\end{equation}
We begin with a preliminary result provided in the following lemma.

\begin{lemma}\label{lem6}
	Under Assumptions $\ref{assum1}$-$\ref{assum4}$, the parameter estimates given by  Algorithm $1$ satisfies that 
	\begin{equation}\label{22}
		\begin{aligned}
\tilde{\theta}_{T+1}^{\top}P_{T+1}^{-1}\tilde{\theta}_{T+1}+\sum_{t=0}^{T}\eta_{t}\psi_{t}^{2}= O\left(\log r_{T}\right),\;\; a.s., 
		\end{aligned}
	\end{equation}
as $T \rightarrow \infty,$ where $\tilde{\theta}_{t}=\theta^{*}-\hat{\theta}_{t}$.
\end{lemma}

\begin{proof} 
	Following the analysis ideas of the classical least-squares for linear stochastic regression models (see e.g., \cite{moore1978}, \cite{lw1982}, \cite{g1995}), we consider the following stochastic  Lyapunov function: $$ V_{t+1}=\tilde{\theta}_{t+1}^{\top}P_{t+1}^{-1}\tilde{\theta}_{t+1}. $$
By (\ref{e27}), we have
\begin{equation}\label{e28}
\begin{aligned}
y_{t+1}-h(\hat{\theta}_{t}, y_{t}, u_{t})
=\psi_{t}+w_{t+1}.
\end{aligned}
\end{equation}
Hence, by \eqref{eq6} and \eqref{e28}, we have for each $t \geq 0$, 
\begin{equation}\label{e555}
\begin{aligned}
V_{t+1}\leq& \left[\tilde{\theta}_{t}-\eta_{t}P_{t+1}\phi_{t}(y_{t+1}-h(\hat{\theta}_{t}, y_{t}, u_{t}))\right]^{\top}P_{t+1}^{-1}\\
&\left[\tilde{\theta}_{t}-\eta_{t}P_{t+1}\phi_{t}(y_{t+1}-h(\hat{\theta}_{t}, y_{t}, u_{t}))\right]\\
=&\tilde{\theta}_{t}^{\top}P_{t+1}^{-1}\tilde{\theta}_{t}-2\eta_{t}\tilde{\theta}_{t}^{\top}\phi_{t}\psi_{t}+\eta_{t}^{2}\psi_{t}^{\top}\phi_{t}^{\top}P_{t+1}\phi_{t}\psi_{t}\\		
		&+2\eta_{t}^{2}\psi_{t}^{\top} \phi_{t}^{\top}P_{t+1}\phi_{t}w_{t+1}-2\eta_{t}\phi_{t}^{\top}\tilde{\theta}_{t}w_{t+1}\\
&+\eta_{t}^{2}w_{t+1}^{\top}\phi_{t}^{\top}P_{t+1}\phi_{t}w_{t+1}.
\end{aligned}
\end{equation}
Moreover, note that by $(\ref{eq6})$ and the well-known matrix inversion formula, we have
\begin{equation}\label{e5566}
P_{t+1}^{-1}=P_{t}^{-1}+\eta_{t}^{2}\phi_{t}\phi_{t}^{\top}.
\end{equation}
Furthermore, since
\begin{equation}
\nabla_{\theta}\mathcal{L}(\theta, y_{t}, u_{t})=-\phi_{t}\psi_{t},
\end{equation}
it follows from Assumption \ref{assum4} that
\begin{equation}\label{e32}
\tilde{\theta}_{t}^{\top}\phi_{t}\psi_{t}\geq \alpha_{t}\tilde{\theta}_{t}^{\top}\phi_{t}\phi_{t}^{\top}\tilde{\theta}_{t}
\end{equation}
and 
\begin{equation}\label{eq33}
\psi_{t}^{\top}\psi_{t}\leq  \beta\tilde{\theta}_{t}^{\top}\phi_{t}\psi_{t}.
\end{equation}
Besides, by the definition of $\eta_{t}$ in $(\ref{eq6})$, we have
\begin{equation}\label{e33}
\;\; \eta_{t}\leq \alpha_{t},\;\; \beta\eta_{t}\left\|\phi_{t}^{\top}P_{t+1}\phi_{t}\right\|\leq \frac{1}{2}.
\end{equation}
Thus, by $(\ref{e32})$-$(\ref{e33})$, we have
\begin{equation}
\eta_{t}\tilde{\theta}_{t}^{\top}\phi_{t}\psi_{t}\geq \eta_{t}^{2}\tilde{\theta}_{t}^{\top}\phi_{t}\phi_{t}^{\top}\tilde{\theta}_{t},
\end{equation}
and
\begin{equation}
\eta_{t}^{2}\psi_{t}^{\top}\phi_{t}^{\top}P_{t+1}\phi_{t}\psi_{t}\leq \frac{1}{2}\eta_{t}\tilde{\theta}_{t}^{\top}\phi_{t}\psi_{t}.
\end{equation}
By $(\ref{e555})$-$(\ref{eq33})$, we can obtain
\begin{equation}\label{27}
	\begin{aligned}
		V_{t+1}
		\leq&V_{t}
-\frac{1}{2}\eta_{t}\tilde{\theta}_{t}^{\top}\phi_{t}\psi_{t}+2\eta_{t}^{2}\psi_{t}^{\top} \phi_{t}^{\top}P_{t+1}\phi_{t}w_{t+1}\\		
		&-2\eta_{t}\phi_{t}^{\top}\tilde{\theta}_{t}w_{t+1}+\eta_{t}^{2}w_{t+1}^{\top}\phi_{t}^{\top}P_{t+1}\phi_{t}w_{t+1}.
	\end{aligned}
\end{equation}
Moreover, by (\ref{eq33}), we have
\begin{equation}\label{e37}
	\begin{aligned}
\frac{1}{4}\eta_{t}\tilde{\theta}_{t}^{\top}\phi_{t}\psi_{t} \geq \frac{1}{4\beta}\eta_{t}\psi_{t}^{\top}\psi_{t}.
	\end{aligned}
\end{equation}
Now, substituting $(\ref{e37})$ into $(\ref{27})$, and
summing up both sides of $\left(\ref{27} \right)$ from $t=0$ to $T$, we obtain
\begin{equation}\label{277}
	\begin{aligned}
		&V_{T+1}+\frac{1}{4\beta}\sum_{t=0}^{T}\eta_{t}\psi_{t}^{\top}\psi_{t}+\frac{1}{4}\sum_{t=0}^{T}\eta_{t}\tilde{\theta}_{t}^{\top}\phi_{t}\psi_{t}\\
		\leq&V_{0}+
\sum_{t=0}^{T}\eta_{t}^{2}\|\phi_{t}^{\top}P_{t+1}\phi_{t}\|\mathbb{E}\left[\|w_{t+1}\|^{2}\mid \mathcal{F}_{t}\right]\\		&-\sum_{t=0}^{T}2\eta_{t}\phi_{t}^{\top}\tilde{\theta}_{t}w_{t+1}+\sum_{t=0}^{T}2\eta_{t}^{2}\psi_{t}^{\top}\phi_{t}^{\top}P_{t+1}\phi_{t}w_{t+1}\\
&+\sum_{t=0}^{T}\eta_{t}^{2}\|\phi_{t}^{\top}P_{t+1}\phi_{t}\|\left(\|w_{t+1}\|^{2}-\mathbb{E}\left[\|w_{t+1}\|^{2}\mid \mathcal{F}_{t}\right]\right).
	\end{aligned}
\end{equation}

We now analyze noise terms on the right-hand side (RHS) of $(\ref{277})$ term by term. 
For the first noise term of $(\ref{277})$, by Lemma \ref{lem3} in Appendix, we have
\begin{equation}\label{a136}
\begin{aligned}
&\sum_{t=0}^{T}\|\eta_{t}^{2}\phi_{t}^{\top}P_{t+1}\phi_{t}\|
=O\left(\log r_{T}\right),\;\;a.s.
\end{aligned}
\end{equation}
Hence, by Assumption $\ref{assum2}$, we have
\begin{equation}\label{a136}
\begin{aligned}
&\sum_{t=0}^{T}\|\eta_{t}^{2}\phi_{t}^{\top}P_{t+1}\phi_{t}\|\|\mathbb{E}\left[\|w_{t+1}\|^{2}\mid \mathcal{F}_{t}\right]
=O\left(\log r_{T}\right),\;\;a.s.
\end{aligned}
\end{equation}
For the last three terms on the RHS of $(\ref{277})$, by Lemma \ref{lem2} in Appendix, we have
\begin{equation}
\begin{aligned}
\sum_{t=0}^{T}\eta_{t}\phi_{t}^{\top}\tilde{\theta}_{t}w_{t+1}
=&o\left(\sum_{t=0}^{T}\eta_{t}^{2}\left[\tilde{\theta}_{t}^{\top}\phi_{t}\phi_{t}^{\top}\tilde{\theta}_{t}\right]\right)\\
=&o\left(\sum_{t=0}^{T}\eta_{t}\left(\phi_{t}^{\top}\tilde{\theta}_{t}\right)\psi_{t}\right),\;\;a.s.
\end{aligned}
\end{equation} 
Moreover, let $s_{T}=\left[\sum\limits_{t=0}^{T}\left\|\eta_{t}^{2}\phi_{t}^{\top}P_{t+1}\phi_{t}\psi_{t}\right\|^{2}\right]^{\frac{1}{2}}$, then by $(\ref{e33})$, we have
\begin{equation}\label{e42}
\begin{aligned}
s_{T}=&O\left(\left[\sum_{t=0}^{T}\eta_{t}^{2}\|\psi_{t}\|^{2}\right]^{\frac{1}{2}}\right)=O\left(\left[\sum_{k=0}^{T}\eta_{t}\|\psi_{t}\|^{2}\right]^{\frac{1}{2}}\right),\;\;a.s.
\end{aligned}
\end{equation}
Thus, by (\ref{e42}) and Lemma \ref{lem2}, we know that
 \begin{equation}
 \begin{aligned}
 &\sum_{t=0}^{T}\eta_{t}^{2}\psi_{t}\phi_{t}^{\top}P_{t+1}\phi_{t}w_{t+1}\\
 =&O\left(s_{T}\sqrt{\log s_{T}}\right)=o\left(\sum_{t=0}^{T}\eta_{t}\|\psi_{t}\|^{2}\right),\;\;a.s.
 \end{aligned}
 \end{equation}
 For the last term on the RHS of $(\ref{277})$, let $$g_{t}=\left[\sum_{i=0}^{t}\left(\eta_{i}^{2}\|\phi_{i}^{\top}P_{i+1}\phi_{i}\|\right)^{\frac{\eta}{2}}\right]^{\frac{2}{\eta}}$$
 Then, by Assumption \ref{assum2}, Lemma \ref{lem2}, and $(\ref{a136})$, we have
\begin{equation}\label{e141}
\begin{aligned}
 &\sum_{t=0}^{T}\eta_{t}^{2}\|\phi_{t}^{\top}P_{t+1}\phi_{t}\|\left(\|w_{t+1}^{2}\|-\mathbb{E}[\|w_{t+1}\|^{2}\mid \mathcal{F}_{t}]\right)\\
 =&O(g_{T}\sqrt{\log g_{T}})=o\left(\log r_{T}\right),\;\;a.s.,
 \end{aligned}
 \end{equation}
where we have used the fact that $\|\eta_{t}^{2}\phi_{t}^{\top}P_{t+1}\phi_{t}\|\leq 1$. Substituting $(\ref{a136})$-$(\ref{e141})$ into $(\ref{277})$, we finally obtain
\begin{equation}\label{e147}
\begin{aligned}
\tilde{\theta}_{T+1}P_{T+1}^{-1}\tilde{\theta}_{T+1}+\sum_{t=0}^{T}\eta_{t}\|\psi_{t}\|^{2}
=O\left(\log r_{T}\right),\;\;a.s.
\end{aligned}
\end{equation}
\end{proof}

By exploiting Lemma~\ref{lem6}, we are ready to prove Theorem \ref{thm2} as follows.

 \begin{proof}[Theorem \ref{thm2}]
 Let $\gamma\in (0, \eta-2)$, $\xi_{t}=\left[y_{t}^{\top}, u_{t}^{\top}\right]^{\top}$, and
$$v_{t+1}=\|w_{t+1}\|^{2+\gamma}-\mathbb{E}\left[\|w_{t+1}\|^{2+\gamma}\mid \mathcal{F}_{t}\right].$$ 
 Then $\{v_{t+1}, \mathcal{F}_{t}\}$ is a martingale difference sequence, and by Lemma \ref{lem2}, we have
 \begin{equation}\label{e46}
\frac{1}{T}\sum_{t=0}^{T-1}v_{t+1}=O\left(\frac{1}{T}m_{T}\sqrt{\log m_{T}}\right)=o(1),\;\;a.s.,
 \end{equation}
 where $m_{t}=t^{\frac{2}{\eta}}$. Thus, from (\ref{e46}), we can obtain
\begin{equation}\label{e44}
\begin{aligned}
&\frac{1}{T}\sum_{t=0}^{T-1}\|w_{t+1}\|^{2+\gamma}\\
=&\frac{1}{T}\sum_{t=0}^{T-1}\mathbb{E}\left[\|w_{t+1}\|^{2+\gamma}\mid \mathcal{F}_{t}\right]+\frac{1}{T}\sum_{t=0}^{T-1}v_{t+1}=O(1),\;\;a.s.
\end{aligned}
\end{equation}
Moreover, from Assumption $\ref{assum3}$ and (\ref{e11}), we have
 \begin{equation}\label{e51}
 \begin{aligned}
 \left\|x_{t+1}\right\|_{1}^{2+\gamma}\leq &\rho \|x_{t}\|_{1}^{2+\gamma}+\|w_{t+1}\|_{1}^{2+\gamma}\\
 \leq &\rho^{t+1} \|x_{0}\|_{1}^{2+\gamma}+\sum_{i=0}^{t}\rho^{t-i}\|w_{i+1}\|_{1}^{2+\gamma},\;\;a.s.
 \end{aligned}
  \end{equation}
 Thus, from $(\ref{e44})$ and $(\ref{e51})$, we have
 \begin{equation}
 \frac{1}{T}\sum_{t=0}^{T-1}\|y_{t}\|^{2+\gamma}=O(1),\;\;a.s.
 \end{equation}
By Assumption $\ref{assum3}$,  we also have
 \begin{equation}
 \frac{1}{T}\sum_{k=0}^{T-1}\|u_{t}\|^{2+\gamma}=O(1),\;\;a.s.
 \end{equation}
Hence, we have 
$$\frac{1}{T}\sum_{t=0}^{T-1}\|\xi_{t}\|^{2+\gamma}=O(1),\;\;a.s.,$$
and
 \begin{equation}\label{e500}
r_{T}=O(T),\;\;a.s.
 \end{equation}
Furthermore, for each $x\in \mathbb{R}^{+}$, define the function $f(\cdot)$ by
$$f(x)=\frac{1}{2\beta\|P_{0}\| M(D)x^{2}+\alpha^{-1}(2x+D)+c_{0}},$$
where $c_{0}$ is a constant defined by $$c_{0}=2\beta\|P_{0}\|[M(D)\|\xi_{0}\|+\sup_{\|s\|\leq D}\|\nabla_{\theta}h(s, x_{0}, u_{0})\|].$$ Then we can easily verify that $f(\|\xi_{t}\|)\leq \eta_{t}$. Thus, by (\ref{e147}) and (\ref{e500}), we  have
\begin{equation}\label{e44}
\sum_{t=0}^{T-1}f(\|\xi_{t}\|)\|\psi_{t}\|^{2}=O(\log T).
\end{equation}
From $(\ref{e44})$, the proof of $(\ref{eqq9})$ can be obtained by along the same lines as those of Theorem 1 in \cite{cdc}, and is omitted here.

To prove \eqref{eqq22}, note that if the sequence $\{y_{t}, t\geq 0\}$ is bounded, then so is $\{\xi_{t}, t\geq 0\}$. Consequently, the sequence $\{f(\|\xi_{t}\|), t\geq 0\}$ has positive lower bound. Therefore, \eqref{eqq22} follows directly from \eqref{e44}.

\end{proof}

\begin{proof}[Theorem \ref{thm1}]
Theorem \ref{thm1} can be directly obtained by Lemma \ref{lem6} and the following fact:
\begin{equation}\label{500i}
\begin{aligned}
\tilde{\theta}_{T+1}^{\top}P_{T+1}^{-1}\tilde{\theta}_{T+1}\geq \lambda_{\min}\left\{P_{T+1}^{-1}\right\}\|\tilde{\theta}_{T+1}\|^{2}.
\end{aligned}
\end{equation}
\end{proof}
 
\section{NUMERICAL EXPERIMENTS}	

Consider the stochastic nonlinear dynamical system
	\begin{equation}\label{a}
\begin{aligned}			x_{t+1}&=\sigma(Ax_{t}+Bu_{t})+w_{t+1}, \\
y_{t+1}&=x_{t+1},
\end{aligned}
	 \end{equation}
    where 
    \begin{equation}
A=\begin{bmatrix} 0.5 & 1\\ 1.5& 0.3 \end{bmatrix}; \;\; B=\begin{bmatrix} 1 \\ 0 \end{bmatrix},
    \end{equation}
     the noise sequence  $\left\{ w_{t}, t\geq 1\right\}$ is i.i.d. with standard normal distribution $\mathcal{N}(0,1)$, the nonlinear function $\sigma(\cdot)$ is defined as in Example \ref{ex1}.  To estimate $\theta=[A, B]^{\top}$ by Algorithm $1$, take the initial value $\hat{\theta}_{0}=[0, 0, 0]^{\top}$, $P_{0}=I$, and $D=2$.

	Let the control input $u_{t}$ be designed by 
    \begin{equation}\label{eu}
    u_{t}=Ky_{t}+\epsilon_{t},
    \end{equation}
    where $K$ is obtained by solving a discrete-time Riccati equation (by setting cost $Q$ and $R$ to identity).
    We now present experiments under the following three cases: 1) $\epsilon_t$ is identically zero. 2) $\epsilon_t=v_{t}$, where $v_{t}$ is i.i.d. noise uniformly distributed on the unit sphere. 3) $\epsilon_t$ is decaying independent noise, i.e., $\epsilon_t = t^{-\frac{1}{4}}\sqrt{\log t} v_{t}$. 

As stated in Example \ref{ex1}, the system dynamics \eqref{a} and control policy \eqref{eu} satisfy Assumptions \ref{assum1}-\ref{assum4}. Theorem \ref{thm2} consequently guarantees asymptotic convergence of the average regret to zero in all three considered cases, as empirically validated by the convergence trajectories of $\frac{1}{t} R_t$ depicted in Fig. 2.
 
 Fig. $3$ illustrates the trajectories of the parameter estimation error. As stated in Theorem \ref{thm1}, the asymptotic upper bound of the estimation error depends on the excitation condition of the data. In Case 1 (Fig. 3(a)), insufficient excitation leads to non-zero estimation error. Case 2 (Fig. 3(b)) incorporates the excitation signal $\epsilon_{t}$ following \cite{sattar}, which ensures the closed-loop data satisfies the PE condition and achieves the theoretical convergence rate of $O(\frac{\log T}{T})$. Notably, Case 3 (Fig. 3(c)) demonstrates that our proposed Algorithm 1 maintains convergence to zero estimation error despite the lack of PE condition in the closed-loop data, highlighting a key advantage of our approach.
 
\begin{figure}[htbp]
	\centering	\subfloat[$\epsilon_{t}=0$]{\includegraphics[width=.9\columnwidth]{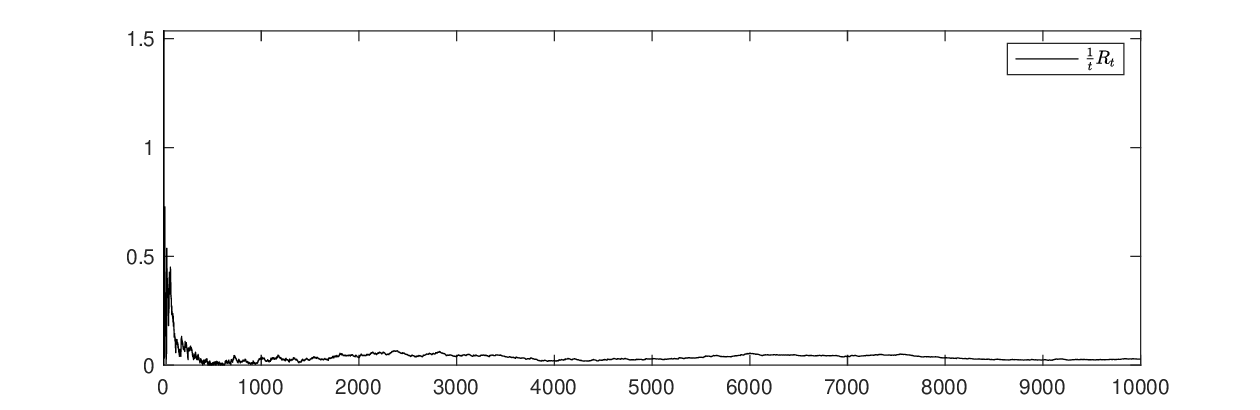}}\\	\subfloat[$\epsilon_{t}=v_{t}$]{\includegraphics[width=.9\columnwidth]{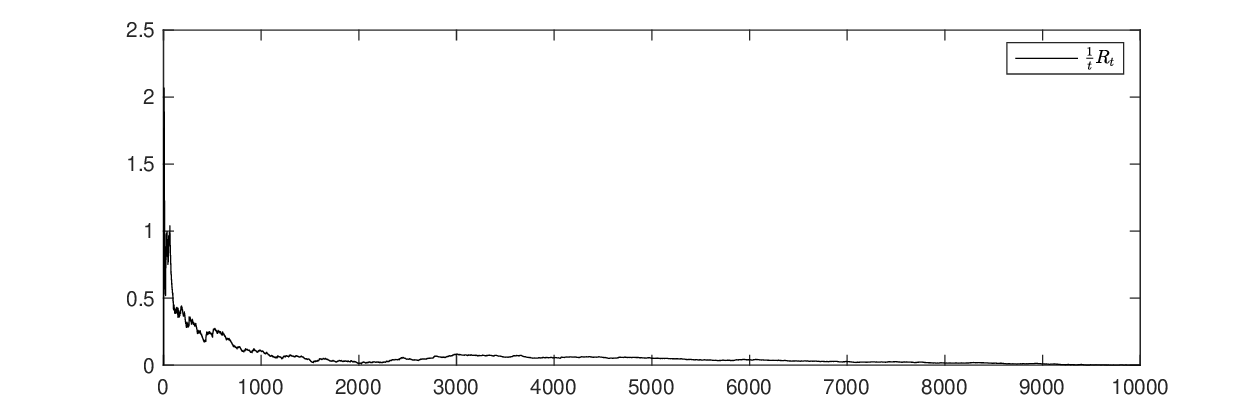}}\\	\subfloat[$\epsilon_{t}=t^{-\frac{1}{4}}\sqrt{\log t} v_{t}$]{\includegraphics[width=.9\columnwidth]{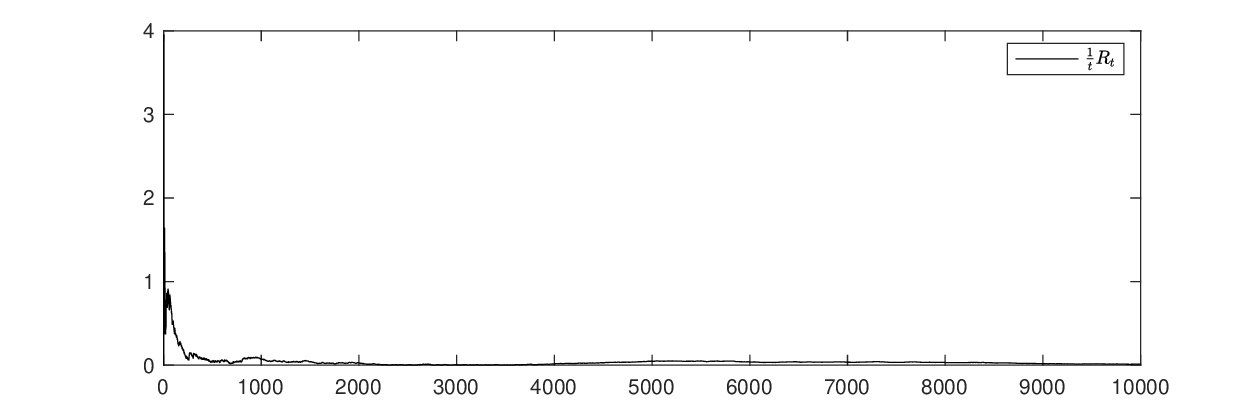}}
	\caption{Trajectories of $\frac{1}{t}R_{t}$}
\end{figure}

\begin{figure}[htbp]
	\centering
	\subfloat[$\epsilon_{t}=0$]
{\includegraphics[width=.9\columnwidth]{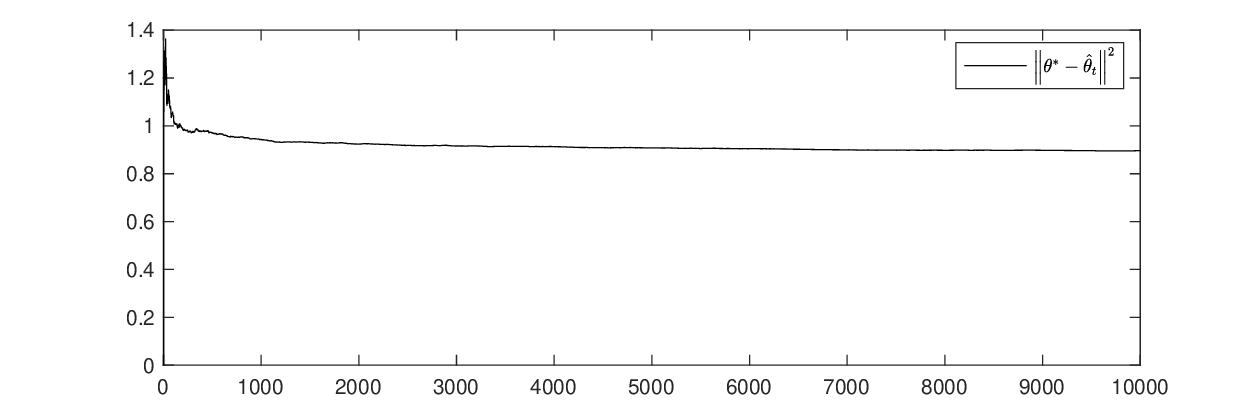}}\\	\subfloat[$\epsilon_{t}=v_{t}$]{\includegraphics[width=.9\columnwidth]{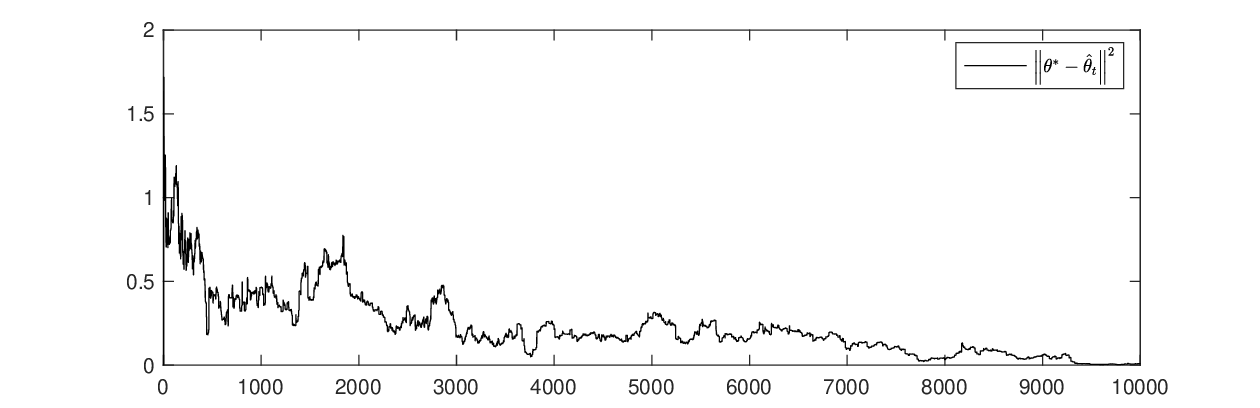}}\\	\subfloat[$\epsilon_{t}=t^{-\frac{1}{4}}\sqrt{\log t} v_{t}$]{\includegraphics[width=.9\columnwidth]{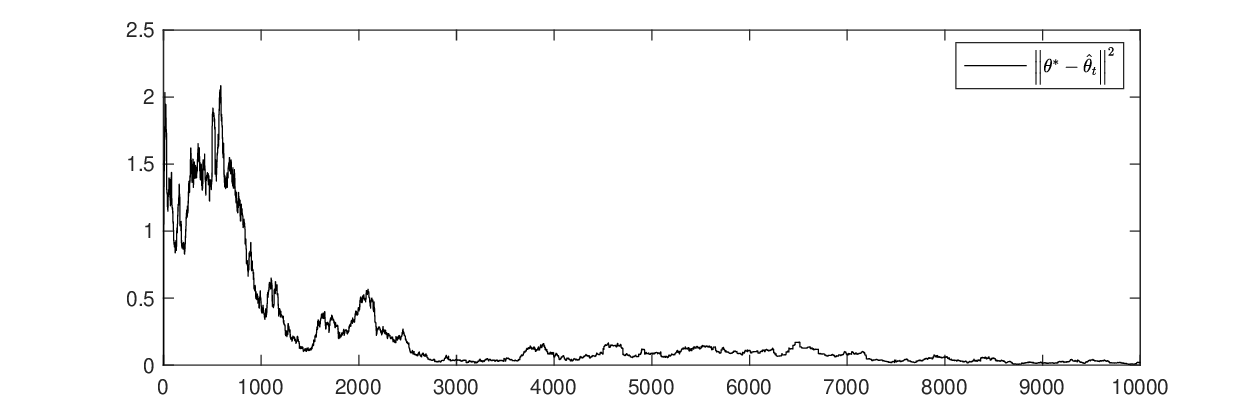}}
	\caption{Trajectories of $\|\theta^{*}-\hat{\theta}_{t}\|^{2}$}
\end{figure}

\section{CONCLUSIONS}
This paper has investigated online identification and prediction for nonlinear stochastic dynamical systems, proposing a novel online learning algorithm for parameter estimation. We have demonstrated that: 1) the average regret of adaptive prediction converges to zero without requiring any excitation conditions, and 2) the parameter estimates converge almost surely to the true values without traditional persistent excitation requirements.  However, there are still many interesting problems that need to be further investigated, for example, 
to design and analyze the online learning algorithms based on partial observations, and to consider the adaptive control problem for nonlinear stochastic dynamical systems.

\addtolength{\textheight}{-1cm}   % This command serves to balance the column lengths
                                  % on the last page of the document manually. It shortens
                                  % the textheight of the last page by a suitable amount.
                                  % This command does not take effect until the next page
                                  % so it should come on the page before the last. Make
                                  % sure that you do not shorten the textheight too much.

%%%%%%%%%%%%%%%%%%%%%%%%%%%%%%%%%%%%%%%%%%%%%%%%%%%%%%%%%%%%%%%%%%%%%%%%%%%%%%%%

%%%%%%%%%%%%%%%%%%%%%%%%%%%%%%%%%%%%%%%%%%%%%%%%%%%%%%%%%%%%%%%%%%%%%%%%%%%%%%%%

%%%%%%%%%%%%%%%%%%%%%%%%%%%%%%%%%%%%%%%%%%%%%%%%%%%%%%%%%%%%%%%%%%%%%%%%%%%%%%%%
\section*{APPENDIX}\label{appendix}

\begin{lemma}\label{lem2} (\cite{cg1991}). Let $\left\{f_{t}, \mathcal{F}_{t}\right\}$ be an adapted sequence and $\left\{w_{t}, \mathcal{F}_{t}\right\}$ a martingale difference sequence.
	If
	\begin{equation}\label{18}
		\sup_{t} 	\mathbb{E}[\|w_{t+1}\|^{\alpha}\mid \mathcal{F}_{t}] < \infty \;\;a.s.
	\end{equation}
	for some $\alpha \in (0, 2]$, then as $T\rightarrow \infty$:
	\begin{equation}\label{19}
		\sum_{t=0}^{T-1}f_{t}w_{t+1} = O\left(s_{T}(\alpha)\log^{\frac{1}{\alpha}+\eta}(s_{T}^{\alpha}(\alpha)+e)\right)\;a.s., \forall \eta >0,
	\end{equation}
	where
	\begin{equation}
		s_{T}(\alpha)=\left(\sum_{t=0}^{T-1}\|f_{t}\|^{\alpha}\right)^{\frac{1}{\alpha}}.
	\end{equation}
\end{lemma}

\begin{lemma}(\cite{lw1982})\label{lem33} Let \;$X_{1}, X_{2},\cdots$ be a sequence of vectors in $	\mathbb{R}^{n} (n\geq 1)$ and let $A_{t} = A_{0}+\sum\limits_{i=1}^{t}X_{i}X_{i}^{\top}$, $b_{t}=\sum\limits_{i=0}^{t}\|X_{i}\|^{2}$. Assume that $A_{0}$ is non-singular, then for any $\delta>0$, 
		\begin{equation}\label{20}
			\sum_{t=0}^{T-1}X_{t}^{\top}A_{t+1}^{-1}X_{t} = O(\log b_{T}),\;\;n\rightarrow \infty.	
		\end{equation}
\end{lemma}

\begin{lemma}\label{lem3} The gradient sequence $\{\phi_{t}, t\geq 0\}$ satisfies the following property as $T\rightarrow \infty$, 
		\begin{equation}\label{20}
			\sum_{t=0}^{T-1}\|\eta_{t}^{2}\phi_{t}^{\top}P_{t+1}\phi_{t}\| = O(\log r_{T}),	
		\end{equation}
        where $\phi_{t}$ and $r_{T}$ are defined in (\ref{e27}).
\end{lemma}
\begin{proof}
Let
$$h(\hat{\theta}_{t}, y_{t}, u_{t})=[h_{1}(\hat{\theta}_{t}, y_{t}, u_{t}), \cdots, h_{n}(\hat{\theta}_{t}, y_{t}, u_{t})]^{\top},$$
and
\begin{equation}
\begin{aligned}
P_{t+1,i}^{-1}&=P_{t,i}^{-1}+\eta_{t}^{2}\nabla_{\theta}h_{i}(\hat{\theta}_{t}, y_{t}, u_{t})\nabla_{\theta}^{\top}h_{i}(\hat{\theta}_{t}, y_{t}, u_{t}),\; t\geq 0,\\
P_{0,i}&=I,\; 1\leq i \leq n. 
\end{aligned}
\end{equation}
where $h_{i}(\theta, y_{t}, u_{t})$ is the $i^{th}$ component of the vector $h(\theta, y_{t}, u_{t})$. Then, we have
$$\phi_{t}=[\nabla_{\theta}h_{1}(\hat{\theta}_{t}, y_{t}, u_{t}), \cdots, \nabla_{\theta}h_{n}(\hat{\theta}_{t}, y_{t}, u_{t})].$$
and
$$P_{t}\leq P_{t,i},\;\; 1\leq i \leq n,\;\; t\geq 0.$$
Thus, we obtain
\begin{equation}\label{e67}
\begin{aligned}
&\sum_{t=0}^{T-1}\|\eta_{t}^{2}\phi_{t}^{\top}P_{t+1}\phi_{t}\|\\
\leq &\sum_{t=0}^{T-1}tr(\eta_{t}^{2}\phi_{t}^{\top}P_{t+1}\phi_{t})\\
=&\sum_{t=0}^{T-1}\sum_{i=1}^{n}\eta_{t}^{2}\nabla_{\theta}^{\top}h_{i}(\hat{\theta}_{t}, y_{t}, u_{t})P_{t+1}\nabla_{\theta}h_{i}(\hat{\theta}_{t}, y_{t}, u_{t})\\
\leq&\sum_{t=0}^{T-1}\sum_{i=1}^{n}\eta_{t}^{2}\nabla_{\theta}^{\top}h_{i}(\hat{\theta}_{t}, y_{t}, u_{t})P_{t+1, i}\nabla_{\theta}h_{i}(\hat{\theta}_{t}, y_{t}, u_{t}).
\end{aligned}
\end{equation}
From (\ref{e67}) and Lemma \ref{lem33}, (\ref{20}) can be obtained.
\end{proof}

\begin{proof}[Example \ref{ex1}]
In this case, the loss function defined in (\ref{eq7}) is  
\begin{equation}
\begin{aligned}
\mathcal{L}(\theta, y_{t}, u_{t})=\left\|\sigma(A^{*}y_{t}+B^{*}u_{t})-\sigma(Ay_{t}+Bu_{t})\right\|^{2}.
\end{aligned}
\end{equation}
Let $\sigma'_{i}(a,b)=\frac{\sigma_{i}(a)-\sigma_{i}(b)}{a-b}$ for each $a, b \in \mathbb{R}, a\not=b$, where $\sigma_{i}(\cdot)$ is defined in  $(\ref{sigma})$. Then, for any $x=[x^{(1)}, \cdots, x^{(n)}]^{\top}, y=[y^{(1)}, \cdots, y^{(n)}]^{\top}\in \mathbb{R}^{n}$, define $\sigma'(x,y)=[\sigma_{1}'(x^{(1)}, y^{(1)}),\cdots, \sigma_{n}'(x^{(n)}, y^{(n)})]^{\top}$. Let $J_{t}=\diag(\sigma'(A^{*}y_{t}+B^{*}u_{t}, Ay_{t}+Bu_{t})).$ Then, $\psi_{t}$, as defined in (\ref{e27}), is
\begin{equation}
\begin{aligned}
\psi_{t}=\sigma(A^{*}y_{t}+B^{*}u_{t})-\sigma(Ay_{t}+Bu_{t})=\langle J_{t}\otimes\begin{bmatrix} y_{t}\\ u_{t} \end{bmatrix}, \theta^{*}-\theta\rangle,
\end{aligned}
\end{equation}
% \sigma'
where $\otimes$ is the Kronecker
product. Besides, the gradient of loss $\mathcal{L}(\theta, y_{t}, u_{t})$ is 
\begin{equation}
\begin{aligned}
\nabla_{\theta} \mathcal{L}(\theta, y_{t}, u_{t})=-2\left[\diag(\sigma'(Ay_{t}+Bu_{t}))\psi_{t}\right]\otimes\begin{bmatrix} y_{t} \\ u_{t} \end{bmatrix}.
\end{aligned}
\end{equation}
Thus, one can obtain
\begin{equation}
\begin{aligned}
&\langle\theta-\theta^{*}, \nabla_{\theta} \mathcal{L}(\theta, y_{t}, u_{t})\rangle\\
=&2(\theta-\theta^{*})^{\top}\left(\left[\diag(\sigma'(Ay_{t}+Bu_{t}))J_{t}\right]\otimes\begin{bmatrix} y_{t}\\ u_{t} \end{bmatrix}\begin{bmatrix} y_{t}^{\top}, u_{t}^{\top} \end{bmatrix}\right)(\theta-\theta^{*}).
\end{aligned}
\end{equation}
Moreover, the gradient of $h(\theta, y_{t}, u_{t})$ is 
\begin{equation}
\nabla_{\theta} \sigma(Ay_{t}+Bu_{t})=\diag(\sigma'(Ay_{t}+Bu_{t}))\otimes\begin{bmatrix} y_{t} \\ u_{t} \end{bmatrix}.
\end{equation}
Let $\alpha(r)=2\sigma'(2r^{2}).$
Since
\begin{equation}
\begin{aligned}
&2\sigma'(A_{i}^{*}y_{t}+B_{i}^{*}u_{t}, A_{i}y_{t}+B_{i}u_{t})\sigma'(A_{i}y_{t}+B_{i}u_{t})\\
\geq& \alpha(r)[\sigma'(A_{i}y_{t}+B_{i}u_{t})]^{2},
\end{aligned}
\end{equation}
where $A_i$ and $B_i$ denote the $i$-th row of the matrices $A$ and $B$, respectively, for $1 \leq i \leq n$, one can obtain
\begin{equation}
\begin{aligned}
&\langle\theta-\theta^{*}, \nabla_{\theta} \mathcal{L}(\theta, y_{t}, u_{t})\rangle\\
\geq &\alpha(r)(\theta-\theta^{*})^{\top}\nabla_{\theta} \sigma(Ay_{t}+Bu_{t})\nabla_{\theta}^{\top} \sigma(Ay_{t}+Bu_{t})(\theta-\theta^{*}).
\end{aligned}
\end{equation}

The proof of $\beta=1$ and $M(r)\equiv 1$ follow directly from the fact $|\sigma_{i}'(x)|\leq 1, \forall x \in \mathbb{R}, 1\leq i \leq n.$

The proof that Example \ref{ex3} satisfies Assumption \ref{assum4} can be given similarly to that of Example \ref{ex1}.
\end{proof}
%%%%%%%%%%%%%%%%%%%%%%%%%%%%%%%%%%%%%%%%%%%%%%%%%%%%%%%%%%%%%%%%%%%%%%%%%%%%%%%%

%\nobalance

\end{document}